\documentclass[a4paper,fleqn,11pt]{article}

\usepackage{amsmath}
\usepackage{amsthm}
\usepackage{amssymb}

\usepackage[a4paper,top=3cm, bottom=3cm, left=3cm, right=3cm]{geometry}
\usepackage[shortlabels]{enumitem}
\usepackage[pdftex,ocgcolorlinks,pagebackref=false]{hyperref}
\usepackage[affil-it]{authblk}

\setlist[enumerate,1]{label={(\roman*)}}

\usepackage[capitalise,nameinlink,noabbrev]{cleveref}

\Crefname{property}{Property}{Properties}

\theoremstyle{plain}
\newtheorem{theorem}{Theorem}[section]
\newtheorem{lemma}[theorem]{Lemma}
\newtheorem{proposition}[theorem]{Proposition}
\newtheorem{corollary}[theorem]{Corollary}

\theoremstyle{definition}
\newtheorem{definition}[theorem]{Definition}

\newcommand{\reals}{\mathbb{R}}
\newcommand{\complexes}{\mathbb{C}}
\newcommand{\naturals}{\mathbb{N}}
\newcommand{\positiveintegers}{\mathbb{N}_{>0}}
\newcommand{\nonnegativereals}{\mathbb{R}_{\ge 0}}
\newcommand{\finitefield}[1]{\mathbb{F}_{#1}}

\newcommand{\distributions}[1][]{\mathcal{P}_{#1}}

\newcommand{\entropy}{H}

\newcommand{\typeclass}[2]{T^{#1}_{#2}}

\newcommand{\graphs}{\mathcal{G}}
\newcommand{\ncgraphs}{\mathcal{G}_{\textnormal{nc}}}
\newcommand{\graphnc}[1]{S_{#1}}
\newcommand{\idealclassical}[1]{\overline{\mathcal{K}_{#1}}}
\newcommand{\idealquantum}[1]{\mathcal{I}_{#1}}

\newcommand{\typegraph}[3]{#1^{\strongproduct #2}[\typeclass{#2}{#3}]}
\renewcommand{\complement}[1]{\overline{#1}}

\newcommand{\strongproduct}{\boxtimes}

\newcommand{\adjacentorequal}{\simeq}
\newcommand{\adjacent}{\sim}

\newcommand{\preorderle}{\preccurlyeq}

\newcommand{\asymptoticle}{\lesssim}

\newcommand{\ket}[1]{\left|#1\right\rangle}
\newcommand{\bra}[1]{\left\langle #1\right|}
\newcommand{\ketbra}[2]{\left|#1\middle\rangle\!\middle\langle#2\right|}
\newcommand{\braket}[2]{\left\langle#1\middle|#2\right\rangle}

\DeclareMathOperator{\boundeds}{\mathcal{B}}

\DeclareMathOperator{\spannedsubspace}{span}

\newcommand{\norm}[2][]{\left\|#2\right\|_{#1}}
\newcommand{\setbuild}[2]{\left\{#1\middle|#2\right\}}

\DeclareMathOperator{\ev}{ev}

\title{Noncommutative extensions of parameters in the asymptotic spectrum of graphs}
\author[1,2]{P\'eter Vrana}
\affil[1]{Department of Geometry, Institute of Mathematics, Budapest University of Technology and Economics, M\H uegyetem rkp. 3., H-1111 Budapest, Hungary.}
\affil[2]{MTA-BME Lend\"ulet Quantum Information Theory Research Group, M\H uegyetem rkp. 3., H-1111 Budapest, Hungary}

\begin{document}
\maketitle

\begin{abstract}
The zero-error capacity of a classical channel is a parameter of its confusability graph, and is equal to the minimum of the values of graph parameters that are additive under the disjoint union, multiplicative under the strong product, monotone under homomorphisms between the complements, and normalized. We show that any such function either has uncountably many extensions to noncommutative graphs with similar properties, or no such extensions at all. More precisely, we find that every extension has an exponent that characterizes its values on the confusability graphs of identity quantum channels, and the set of admissible exponents is either an unbounded subinterval of $[1,\infty)$ or empty. In particular, the set of admissible exponents for the Lov\'asz number, the projective rank, and the fractional Haemers bound over the complex numbers are maximal, while the fractional clique cover number does not have any extensions.
\end{abstract}

\section{Introduction}

The Shannon (zero-error) capacity \cite{shannon1956zero} and several quantum versions thereof have recently been shown to have dual characterizations as the minimum over parameters of graphs (or noncommutative graphs) which are normalized, multiplicative under the strong (or tensor) product, additive under the disjoint union (direct sum), and monotone with respect to appropriately chosen cohomomorphism preorders \cite{zuiddam2019asymptotic,li2020quantum}. With these operations, isomorphism classes of graphs and noncommutative graphs form the semirings $\graphs$ and $\ncgraphs$, and the inclusion $\graphs\to\ncgraphs$ is a semiring-homomorphism. On $\ncgraphs$ (and, by restriction, on $\graphs$) one introduces the cohomomorphism preorder $\le$ and the entanglement-assisted cohomomorphism preorder $\le_*$. The different notions of one-shot zero-error capacities can be characterized in terms of these relations (e.g. the independence number of a graph $G$ is the largest $d$ such that $\complement{K_d}\le G$), and the corresponding (asymptotic) capacities can be similarly understood in terms of \emph{asymptotic preorders} associated with $\le$ and $\le_*$.

The aforementioned dual characterizations are proved using the theory of asymptotic spectra developed by Strassen in the context of computational complexity of bilinear maps \cite{strassen1988asymptotic}. A central concept here is the \emph{asymptotic spectrum} of a preordered semiring $(S,\preorderle)$, the set of $\preorderle$-monotone semiring-homomorphisms from $S$ to $\nonnegativereals$ (with the usual operations and total order), denoted by $\Delta(S,\preorderle)$. Under some conditions on the preorder $\preorderle$ (see \cref{sec:preliminaries} for the precise statements), the asymptotic preorder $\asymptoticle$ is characterized as $x\asymptoticle y$ iff $\forall f\in\Delta(S,\preorderle):f(x)\le f(y)$. The same conditions imply that the asymptotic spectrum is nonempty and the restriction map to the asymptotic spectrum of any subsemiring is surjective.

In \cite{zuiddam2019asymptotic} Zuiddam showed that these conditions are satisfied by $(\graphs,\le)$, and in \cite{li2020quantum} Li and Zuiddam showed that they are also satisfied by $(\graphs,\le_*)$ and $(\ncgraphs,\le_*)$ but not by $(\ncgraphs,\le)$. For this reason, the results of Strassen's theory (the duality theorem, nonempty spectrum, extension property) are not directly available for the study of unassisted (classical or quantum) capacity of quantum channels, which raises the following questions: Is $\Delta(\ncgraphs,\le)$ nonempty? Is the restriction map $\Delta(\ncgraphs,\le)\to\Delta(\graphs,\le)$ surjective? Does $\Delta(\ncgraphs,\le)$ characterize the asymptotic preorder on $\ncgraphs$ (and therefore also the classical and quantum unassisted capacities)?

The first question is known to have a positive answer since $\Delta(\ncgraphs,\le_*)\subseteq\Delta(\ncgraphs,\le)$ is nonempty (in fact, an explicit element is known: a quantum version of the Lov\'asz number \cite{duan2012zero}). Focusing on the second question, in this paper, we study the restriction map from the asymptotic spectrum of noncommutative graphs with respect to the (unassisted) cohomomorphism preorder to the asymptotic spectrum of graphs. Known elements of $\Delta(\graphs,\le)$ include the fractional clique cover number \cite{shannon1956zero}, the Lov\'asz number \cite{lovasz1979shannon}, the projective rank \cite{mancinska2016quantum}, and the fractional Haemers bounds \cite{blasiak2013graph,bukh2018fractional}, therefore we are studying the existence of noncommutative extensions of these and similar parameters.

Formulated in the abstract framework of preordered semirings, Strassen's results and their recent generalizations are fundamentally non-constructive, therefore we cannot use them to find explicit extensions of the graph parameters. However, they are useful for proving that (multiplicative, additive, and monotone) extensions of certain parameters exist and, as we will see, provide information about the possible extensions.

\paragraph{Results.}
None of the known extension theorems apply to the inclusion of $\graphs$ into $\ncgraphs$ with the unassisted cohomomorphism preorder. To overcome this difficulty, we introduce an intermediate semiring $\mathcal{A}$ that contains the confusability graphs of \emph{all} classical channels and the confusability graphs of \emph{noiseless} quantum channels. It follows from the construction of $\mathcal{A}$ that every element in the asymptotic spectrum of $\mathcal{A}$ is the restriction of at least one element of the asymptotic spectrum of noncommutative graphs (with respect to unassisted cohomomorphisms), and that every element in the asymptotic spectrum of $\mathcal{A}$ is uniquely specified by a pair $(f,\alpha)$ where $f$ is in the asymptotic spectrum of graphs and $\alpha\in[1,\infty)$. $\alpha$ is equal to the logarithm of the value of the spectral point on the confusability graph of the noiseless qubit channel, therefore we will call it the \emph{exponent} of the extension. We prove the following about $\Delta(\mathcal{A},\le)$, the asymptotic spectrum of $\mathcal{A}$, considered as a subset of $\Delta(\graphs,\le)\times[1,\infty)$:
\begin{itemize}
\item if $\alpha$ is an admissible exponent for some $f\in\Delta(\graphs,\le)$ (i.e. $(f,\alpha)\in\Delta(\mathcal{A},\le)$) and $\alpha\le\beta$, then $\beta$ is also admissible (\cref{thm:exponentsupperset});
\item $\Delta(\mathcal{A},\le)$ is log-convex (\cref{thm:Aspectrumlogconvex}).
\end{itemize}

For specific graph parameters we find the following conditions concerning the existence of (multiplicative, additive, monotone) extensions:
\begin{itemize}
\item the fractional clique cover number $\complement{\chi}_f$ has no such extensions (\cref{prop:fccnoextensions});
\item the Lov\'asz number $\vartheta$, the fractional Haemers bound over the complex numbers $\mathcal{H}^\complexes_f$, and the complement of the projective rank $\complement{\xi_f}$ each have extensions with every exponent in $[1,\infty)$ (\cref{prop:thetaoneadmissible,prop:fHaemersConeadmissible,prop:complementprojectiverankoneadmissible});
\item for sufficiently large primes $p$ such that there exists a Hadamard matrix of size $4p$, the set of admissible exponents is a proper subset of $[1,\infty)$ (\cref{prop:fHaemersFpbound}).
\end{itemize}

\paragraph{Organization of this paper.}
In \cref{sec:preliminaries} we recall some basic definitions and introduce notations related to graphs and noncommutative graphs in zero-error information theory, and concepts from Strassen's theory of asymptotic spectra, formulated in terms of preordered semirings. In \cref{sec:extensions} we describe an intermediate semiring between the semiring of graphs and the semiring of noncommutative graphs, such that any element of the spectrum of the intermediate semiring admits at least on extension to noncommutative graphs. We study the relation between the asymptotic spectrum of graphs and that of the intermediate semiring and find that the existence of a single extension of any given element of the asymptotic spectrum of graphs implies the existence of an infinite family of extensions. This abstract result is followed by a case-by-case study of the extensions of some known elements of the asymptotic spectrum of graphs. In \cref{sec:convexity} we prove that the asymptotic spectrum of the intermediate semiring is log-convex, which is especially useful for studying the extensions of (suitably defined) convex combinations of elements in the asymptotic spectrum of graphs.

\section{Preliminaries}\label{sec:preliminaries}

\subsection{Graphs and noncommutative graphs.}

By a graph we will mean a finite simple undirected graph. The vertex and edge sets of a graph $G$ will be denoted by $V(G)$ and $E(G)$, and we will write $g\adjacent g'$ if two vertices $g,g'$ are adjacent, and $g\adjacentorequal g'$ if they are adjacent or equal. For example, the \emph{strong product} $G\strongproduct H$ of two graphs has vertex set $V(G)\times V(H)$, and $(g,h)\adjacentorequal(g',h')$ iff $g\adjacentorequal g'$ and $h\adjacentorequal h'$. The \emph{disjoint union} of $G$ and $H$ is the graph $G\sqcup H$ with $V(G\sqcup H)=V(G)\sqcup V(H)$ and $E(G\sqcup H)=E(G)\sqcup E(H)$. The \emph{complete graph} with vertex set $[d]=\{1,2,\ldots,d\}$ is denoted by $K_d$. The complement of $G$ will be denoted by $\complement{G}$. A \emph{homomorphism} $\varphi:H\to G$ is a map between the vertex sets such that adjacent vertices are mapped to adjacent ones, while a \emph{cohomomorphism} is a homomorphism between the complements. We will write $H\le G$ if there exists a cohomomorphism from $H$ to $G$ (i.e. a homomorphism from $\complement{H}$ to $\complement{G}$).

Recall that the \emph{confusability graph} of a classical (or classical-quantum) channel is defined as follows. The vertex set is the set of input symbols, and two distinct vertices form an edge if the corresponding output distributions (or states) do not have disjoint supports. When two independent channels are used in parallel, the transition probabilities are given by the tensor product, and the confusability graph of the product of channels is the strong graph product of the confusability graphs of the individual channels. This implies that the zero-error classical capacity of a channel is a function of its confusability graph, 

A \emph{noncommutative graph} is a subspace $S\subseteq\boundeds(\mathcal{H})$ for a Hilbert-space $\mathcal{H}$ (which we always assume to be finite dimensional) such that $I\in S$ and $S^*=S$ \cite{duan2012zero}. A quantum channel $N:\boundeds(\mathcal{H})\to\boundeds(\mathcal{K})$, with Kraus representation
\begin{equation}
N(\rho)=\sum_{i\in I}E_i\rho E_i^*
\end{equation}
determines the noncommutative graph $\spannedsubspace\setbuild{E_i^*E_j}{i,j\in I}$ --- the \emph{confusability graph} of $N$ --- which, similarly to the classical case, determines its zero-error capacities. Every noncommutative graph arises in this way for a suitable channel \cite{duan2009super,cubitt2011superactivation}. Parallel use of channels with confusability graphs $S\subseteq\boundeds(\mathcal{H})$ and $T\subseteq\boundeds(\mathcal{K})$ corresponds to the \emph{tensor product} of the operator systems $S\otimes T=\spannedsubspace\setbuild{A\otimes B}{A\in S,B\in T}\subseteq\boundeds(\mathcal{H}\otimes\mathcal{K})$. The analogue of the disjoint union of graphs is the \emph{direct sum} $S\oplus T=\setbuild{A\oplus B}{A\in S,B\in T}\subseteq\boundeds(\mathcal{H}\oplus\mathcal{K})$. The noncommutative graphs $S$ and $T$ are isomorphic if there is a unitary $U:\mathcal{H}\to\mathcal{K}$ such that $USU^*=T$.

A \emph{cohomomorphism} from the noncommutative graph $T$ to $S$ is a collection of linear maps $E_i:\mathcal{K}\to\mathcal{H}$ ($i\in I$ for some finite index set $I$) such that $\sum_{i\in I}E_i^*E_i=I$ and $\forall i,j\in I:E_i^*SE_S\subseteq T$. We will write $T\le S$ when a cohomomorphism exists.

Any classical channel $W$ can be viewed as a quantum channel with Kraus operators of the form $\sqrt{W(y|x)}\ketbra{y}{x}$ for input and output symbols $x$ and $y$. The corresponding noncommutative graph depends only on the confusability graph of the classical channel, therefore any graph $G$ can be viewed as a noncommutative graph, namely
\begin{equation}
\graphnc{G}=\spannedsubspace\setbuild{\ketbra{x}{x'}}{x,x'\in V(G),x\adjacentorequal x'}\subseteq\boundeds(\complexes^{V(G)}).
\end{equation}
This embedding of graphs into noncommutative graphs is compatible with the notions of isomorphism and cohomomorphism, and turns disjoint unions into direct sums and strong products into tensor products.

The noiseless $d$-dimensional quantum channel is the identity map on $\boundeds(\complexes^d)$. We will denote its confusability graph $\complexes I\subseteq\boundeds(\complexes^d)$ by $\idealquantum{d}$. Likewise, the noiseless classical channel with $d$ inputs has confusability graph $\complement{K_d}$. We will use the simplified notation $\idealclassical{d}$ for the corresponding noncommutative graph on $\complexes^d$ (which is the set of diagonal operators in the standard basis). We note that $\idealclassical{1}=\idealquantum{1}$ and that, up to isomorphisms, $\idealclassical{d_1}\oplus\idealclassical{d_2}=\idealclassical{d_1+d_2}$, $\idealclassical{d_1}\otimes\idealclassical{d_2}=\idealclassical{d_1d_2}$, and $\idealquantum{d_1}\otimes\idealquantum{d_2}=\idealquantum{d_1d_2}$.

A zero-error code for a classical channel is the same as an independent set of its confusability graph $G$, i.e. a subset $S\subseteq V(G)$ such that no two vertices in $S$ are adjacent. The independence number $\alpha(G)$ is the maximum cardinality of an independent set. The Shannon capacity of $G$ is the limit $\Theta(G)=\lim_{n\to\infty}\sqrt[n]{\alpha(G^{\strongproduct n})}$. Note that $C_0(G)=\log\Theta(G)$ is often also called the Shannon (or zero-error) capacity. The independence number can be characterized in terms of cohomomorphisms as $\alpha(G)=\max\setbuild{d\in\naturals}{\complement{K}_d\le G}$. The notion of an independent set extends to noncommutative graphs, with the same interpretation (transmitting classical information via a quantum channel without errors), and a similar characterization in terms of cohomomorphisms from $\idealclassical{d}$. In addition, one can consider the quantity $\max\setbuild{d\in\naturals}{\idealquantum{d}\le S}$ and its growth rate on tensor powers, which characterizes the ability of the channel to transmit quantum information without errors.

\subsection{Preordered semirings.}

We will use the following facts from the theory of asymptotic spectra due to Strassen. For more details we refer the reader to \cite{strassen1988asymptotic,zuiddam2018algebraic}.

A (commutative) \emph{preordered semiring} is a set $S$ equipped with binary operations $+$ and $\cdot$ that are commutative and associative, with neutral elements $0$ and $1$, and satisfying the distributive law, together with a reflexive and transitive relation $\preorderle$ that is compatible with the operations in the sense that $a\preorderle b$ implies $a+c\preorderle b+c$ and $ac\preorderle bc$ for all $a,b,c\in S$. There is a unique semiring-homomorphism $i:\naturals\to S$, which we will assume to be an order embedding with respect to the usual total order $\le$ on $\naturals$, i.e. for $n,m\in\naturals$ we have $i(n)\preorderle i(m)$ iff $n\le m$.

Let $S$ and $T$ be preordered semirings with preorders $\preorderle_S$, $\preorderle_T$ (for simplicity, we will use the same notation for the operations and neutral elements). A \emph{monotone semiring homomorphism} from $S$ to $T$ is a map $f:S\to T$ satisfying $f(0)=0$, $f(1)=1$, $f(a+b)=f(a)+f(b)$, $f(a\cdot b)=f(a)f(b)$ and $a\preorderle_S b\implies f(a)\preorderle_T f(b)$. The \emph{asymptotic spectrum} $\Delta(S,\preorderle)$ is the set of monotone semiring homomorphisms $f:S\to\nonnegativereals$, and its elements are also called spectral points.

An element $u\in S$ is \emph{power universal} \cite{fritz2020generalization} if for all $s\in S\setminus\{0\}$ there exists $k\in\naturals$ such that $1\preorderle u^ks$ and $s\preorderle u^k$. If such an element exists, we say that $S$ is of \emph{polynomial growth}. If $u=2$ is power universal, then $\preorderle$ is called a \emph{Strassen preorder} \cite{strassen1988asymptotic,zuiddam2018algebraic}.

A preordered semiring of polynomial growth can be equipped with a relaxed preorder, the \emph{asymptotic preorder} \cite{strassen1988asymptotic,vrana2021generalization}, defined as $a\asymptoticle b$ if there is a sublinear sequence $(k_n)_{n\in\naturals}$ of natural numbers such that for all $n$ the inequality $a^n\preorderle u^{k_n}b^n$ holds. The asymptotic spectrum provides the following dual characterization of the asymptotic preorder of a semiring with a Strassen preorder.
\begin{theorem}[{\cite[Corollary 2.6]{strassen1988asymptotic}}, {\cite[Theorem 2.12]{zuiddam2018algebraic}}]\label{thm:Strassen}
Let $S$ be a semiring with a Strassen preorder $\preorderle$, and let $a,b\in S$. Then $a\asymptoticle b$ iff $\forall f\in\Delta(S,\preorderle):f(a)\le f(b)$. Moreover, $\Delta(S,\preorderle)$ is nonempty.
\end{theorem}
It is possible to weaken the conditions on the preorder, but it should be noted that the polynomial growth condition in itself is not sufficient to ensure a similar characterization of the asymptotic preorder.

When $S$ is a preordered semiring and $T$ is a subsemiring with inclusion $i:T\to S$, any element $f$ of $\Delta(S,\preorderle)$ restricts to an element of $\Delta(T,\preorderle)$, which gives rise to a map $\Delta(i):f\mapsto f\circ i$. It is known that when $\preorderle$ is a Strassen preorder on $S$, then $\Delta(i)$ is surjective (see \cite[Corollary 2.7]{strassen1988asymptotic} and \cite[Corollary 2.17]{zuiddam2018algebraic}), i.e. every monotone homomorphism $T\to\nonnegativereals$ has at least one extension to $S$ that is also a monotone homomorphism. The same is true if $S$ is only assumed to be of polynomial growth and $T$ contains a power universal element.
\begin{proposition}[{\cite[Proposition 3]{vrana2021generalization}} and {\cite[13.16. Theorem]{fritz2020local}}]\label{prop:extension}
Let $S$ be a preordered semiring of polynomial growth, $u\in S$ power universal. Let $T\subseteq S$ be a subsemiring containing $u$, with the restricted preorder, and let $i:T\to S$ be the inclusion map. Then $\Delta(i)$ is surjective.
\end{proposition}

We will consider two particular semirings. The first one, denoted by $\graphs$, is the set of isomorphism classes of graphs, with operations given by the disjoint union and strong product, and the \emph{cohomomorphism preorder}: $H\le G$ iff there exists a homomorphism $\complement{H}\to\complement{G}$. This preordered semiring was introduced by Zuiddam in \cite{zuiddam2019asymptotic}, where it was proved that $\le$ is a Strassen preorder on $\graphs$ and that the zero-error capacity of a graph $G$ is characterized as
\begin{equation}
\Theta(G)=\min_{f\in\Delta(\graphs,\le)}f(G).
\end{equation}
$\Delta(\graphs,\le)$ has uncountably many elements \cite{vrana2021probabilistic}, including several well-known upper bounds on the Shannon capacity: the fractional clique cover number $\complement{\chi_f}$, the Lov\'asz number $\vartheta$, the fractional Haemers bounds $\mathcal{H}^{\mathbb{F}}_f$ and the projective rank of the complement $\complement{\xi_f}$.

The second semiring, $\ncgraphs$, consists of isomorphism classes of noncommutative graphs, equipped with the operations induced by the direct sum and the tensor product, and where $S\le T$ if there is a cohomomorphism $S\to T$. We have seen that any graph $G$ gives rise to a noncommutative graph $\graphnc{G}$ in a way that is compatible with the operations and the preorder, which gives rise to a monotone semiring homomorphism $\graphs\to\ncgraphs$ (in fact, an order embedding \cite{li2020quantum}).

Further quantum variants have been considered in \cite{li2020quantum}, in particular with the preorder replaced with the entanglement-assisted cohomomorphism preorder $\le_*$ \cite{cubitt2014bounds,stahlke2015quantum}. In that paper it was shown that $\le_*$ is a Strassen preroder both on $\graphs$ and $\ncgraphs$, but $\le$ is not a Strassen preorder on $\ncgraphs$. Strassen's theory therefore provides dual characterizations of entanglement-assisted classical capacities of classical and quantum channels (and, more generally, the asymptotic preorder), and relates the asymptotic spectra $\Delta(\graphs,\le_*)$ and $\Delta(\ncgraphs,\le_*)$ via the extension property. In contrast, little is known about $\Delta(\ncgraphs,\le)$ and its relation to the unassisted capacity, apart from the simple fact that $\Delta(\ncgraphs,\le_*)\subseteq\Delta(\ncgraphs,\le)$, which is a consequence of $\le_*\supseteq\le$.

The asymptotic preorder of $\ncgraphs$ characterizes the trade-off for simultaneously transmitting classical and quantum information: a rate pair $(R_c,R_q)$ is achievable via a channel with confusability graph $S$ iff $\idealclassical{2}^{\lfloor R_cn\rfloor}\otimes\idealquantum{2}^{\lfloor R_qn\rfloor}\asymptoticle S^{\otimes n}$ for all $n$. A necessary condition for this is that $\forall f\in\Delta(\ncgraphs,\le)$ the inequality $R_c+R_q\log f(\idealquantum{2})\le\log f(S)$, and if the dual characterization property analogous to \cref{thm:Strassen} was true for $(\ncgraphs,\le)$, then this condition would be sufficient as well.

\section{Extensions of graph parameters}\label{sec:extensions}

The inclusion $\graphs\to\ncgraphs$ does not satisfy the conditions of \cref{prop:extension} with the unassisted cohomomorphism preorder, therefore the restriction map may not be surjective. In this section we will introduce an intermediate semiring $\mathcal{A}$ such that \cref{prop:extension} applies to the inclusion $\mathcal{A}\to\ncgraphs$ and at the same time the extensions of monotone semiring homomorphisms from $\graphs$ to $\mathcal{A}$ can be given fairly explicitly in terms of $\Delta(\graphs,\le)$ and an additional parameter with a range that depends on the classical parameter. It will turn out that this range is always a closed interval, bounded from below and unbounded from above (unless it is empty), therefore the problem of understanding of $\Delta(\mathcal{A},\le)$ is reduced to determining the endpoint of the interval as a function on $\Delta(\graphs,\le)$.

We let $\mathcal{A}$ be the subsemiring of $\ncgraphs$ generated by $\graphs$ and $\setbuild{\idealquantum{d}}{d\in\naturals}$. By construction, $\mathcal{A}$ contains a power universal element of $\ncgraphs$, therefore the restriction map $\Delta(\ncgraphs,\le)\to\Delta(\mathcal{A},\le)$ is surjective. Clearly, any element of $\Delta(\mathcal{A},\le)$ is uniquely specified by its restriction to $\graphs$ (which is an element of $\Delta(\graphs,\le)$) together with its values on $\idealquantum{d}$ for $d\in\naturals$. Since $\idealquantum{d_1}\otimes\idealquantum{d_2}=\idealquantum{d_1d_2}$ and $d_1\le d_2$ implies $\idealquantum{d_1}\le\idealquantum{d_2}$, the composition with the map $d\mapsto\idealquantum{d}$ is a monotone increasing multiplicative map, which necessarily has the form $d\mapsto d^\alpha$ for some exponent $\alpha\ge 0$. To see this, apply the monotone homomorphism to the chain of inequalities
\begin{equation}
\idealquantum{2}^{\otimes\lfloor n\log d\rfloor}\le\idealquantum{d}^{\otimes n}\le\idealquantum{2}^{\otimes\lceil n\log d\rceil}
\end{equation}
and let $n\to\infty$ after raising to the power $1/n$. Moreover, $\idealclassical{d}\le\idealquantum{d}$ implies that the exponent $\alpha$ is at least $1$. This gives an outer approximation of $\Delta(\mathcal{A},\le)$ parametrized by $\Delta(\graphs,\le)\times[1,\infty)$. We now introduce a notation for the homomorphism (not necessarily monotone) corresponding to any such pair.
\begin{definition}\label{def:falpha}
Let $f\in\Delta(\graphs,\le)$ and $\alpha\in[1,\infty)$. We define the homomorphism $f_\alpha:\mathcal{A}\to\nonnegativereals$ as
\begin{equation}
f_\alpha\left(\bigoplus_{d=1}^r\graphnc{G_d}\otimes\idealquantum{d}\right)=\sum_{d=1}^r f(G_d)d^\alpha.
\end{equation}
\end{definition}

By the preceding discussion, every element of $\Delta(\mathcal{A},\le)$ is of the form $f_\alpha$ for a unique $f\in\Delta(\graphs,\le)$ and $\alpha\ge 1$, and the question is therefore to decide which pairs $(f,\alpha)$ give rise to a monotone function. We will say that the exponent $\alpha$ is \emph{admissible} (for $f$) if $f_\alpha\in\Delta(\mathcal{A},\le)$.

Our next aim is to show that increasing the exponent $\alpha$ does not lead out of $\Delta(\mathcal{A},\le)$, i.e. the set of admissible exponents for a given element in the asymptotic spectrum of graphs is upwards closed. First we argue that monotonicity continues to hold for the increased exponent when the left hand side of the inequality is a single term of the form $\graphnc{H}\otimes\idealquantum{q}$, to which the general case will be reduced via a type decomposition.
\begin{lemma}\label{lem:nosmallterms}
Suppose that $\graphnc{H}\otimes\idealquantum{q}\le\bigoplus_{d=1}^r\graphnc{G_d}\otimes\idealquantum{d}$. Then $\graphnc{H}\otimes\idealquantum{q}\le\bigoplus_{d=q}^r\graphnc{G_d}\otimes\idealquantum{d}$.
\end{lemma}
\begin{proof}
The assumption means that there exist linear maps $E_i:\complexes^{V(H)}\otimes\complexes^q\to\bigoplus_{d=1}^r\complexes^{V(G_d)}\otimes\complexes^d$ ($i\in I$) such that $\sum_{i\in I}E_i^*E_i=I$ and for all $i,j\in I$ the containment
\begin{equation}
E_i^*\left(\bigoplus_{d=1}^r\graphnc{G_d}\otimes\idealquantum{d}\right)E_j\subseteq\graphnc{H}\otimes\idealquantum{q}
\end{equation}
holds. Every operator in $\graphnc{H}\otimes\idealquantum{q}$ is of the form $X\otimes I_q$, therefore its rank is a multiple of $q$. If $g\in V(G_d)$, then $\ketbra{g}{g}\otimes I_d$ is an element of the right hand side and has rank $d$, therefore $E_i^*(\ketbra{g}{g}\otimes I_d)E_i\in\graphnc{H}\otimes\idealquantum{q}$ has rank at most $d$, and at the same time divisible by $q$. It follows that if $d<q$ then $E_i^*(\ketbra{g}{g}\otimes I_d)E_i=0$.

In other words, the range of each $E_i$ must be orthogonal to the term $\complexes^{V(G_d)}\otimes\complexes^d$ for all $d<q$. Let $P$ be the orthogonal projection $\bigoplus_{d=1}^r\complexes^{V(G_d)}\otimes\complexes^d\to\bigoplus_{d=q}^r\complexes^{V(G_d)}\otimes\complexes^d$. Then the linear maps $(PE_i)_{i\in I}$ form a cohomomorphism from $\graphnc{H}\otimes\idealquantum{q}$ to $\bigoplus_{d=q}^r\graphnc{G_d}\otimes\idealquantum{d}$.
\end{proof}

\begin{corollary}\label{cor:singletermlhsmonotone}
If $f_\alpha\in\Delta(\mathcal{A},\le)$, $\alpha\le\beta$, and $\graphnc{H}\otimes\idealquantum{q}\le\bigoplus_{d=1}^r\graphnc{G_d}\otimes\idealquantum{d}$ then
\begin{equation}
f_\beta\left(\graphnc{H}\otimes\idealquantum{q}\right)\le f_\beta\left(\bigoplus_{d=1}^r\graphnc{G_d}\otimes\idealquantum{d}\right).
\end{equation}
\end{corollary}
\begin{proof}
By \cref{lem:nosmallterms}, the assumption implies $\graphnc{H}\otimes\idealquantum{q}\le\bigoplus_{d=q}^r\graphnc{G_d}\otimes\idealquantum{d}$. Using the monotonicty of $f_\alpha$ for this inequality, we have
\begin{equation}
\begin{split}
f_\beta\left(\graphnc{H}\otimes\idealquantum{q}\right)
 & = q^{\beta-\alpha}f_\alpha\left(\graphnc{H}\otimes\idealquantum{q}\right)  \\
 & \le q^{\beta-\alpha}f_\alpha\left(\bigoplus_{d=q}^r\graphnc{G_d}\otimes\idealquantum{d}\right)  \\
 & = \sum_{d=q}^r f(G_d)q^{\beta-\alpha}d^\alpha  \\
 & \le \sum_{d=q}^r f(G_d)d^{\beta-\alpha}d^\alpha  \\
 & \le \sum_{d=1}^r f(G_d)d^{\beta}  \\
 & = f_\beta\left(\bigoplus_{d=1}^r\graphnc{G_d}\otimes\idealquantum{d}\right).
\end{split}
\end{equation}
\end{proof}

\begin{theorem}\label{thm:exponentsupperset}
If $f_\alpha\in\Delta(\mathcal{A},\le)$ and $\alpha\le\beta$, then $f_\beta\in\Delta(\mathcal{A},\le)$.
\end{theorem}
\begin{proof}
Let $T=\bigoplus_{d=1}^r\graphnc{H_d}\otimes\idealquantum{d}$ and $S=\bigoplus_{d=1}^r\graphnc{G_d}\otimes\idealquantum{d}$ and suppose that $T\le S$. Then $T^{\otimes n}\le S^{\otimes n}$ for all $n\in\naturals$. For $Q\in\distributions[n]([r])$, let us introduce the noncommutative graph
\begin{equation}
T_Q=\graphnc{\complement{K_{\lvert\typeclass{n}{Q}\rvert}}\strongproduct\prod_{d=1}^r H_d^{\strongproduct nQ(d)}}\otimes\idealquantum{\prod_{d=1}^rd^{nQ(d)}}.
\end{equation}
Note that $T_Q$ is the product of a classical graph and the confusability graph of a noiseless quantum channel, therefore \cref{cor:singletermlhsmonotone} applies to any inequality with $T_Q$ on the left hand side. Note also that $T^{\otimes n}$ is the direct sum of the $T_Q$ for all $Q\in\distributions[n]([r])$, therefore $T_Q\le S^{\otimes n}$. Using \cref{cor:singletermlhsmonotone} and that $\lvert\distributions[n]([r])\rvert\le(n+1)^r$ \cite[Lemma 2.2]{csiszar2011information}, we obtain
\begin{equation}
\begin{split}
f_\beta(T)^n
 & = f_\beta(T^{\otimes n})  \\
 & = \sum_{Q\in\distributions[n]([r])}f_\beta(T_Q)  \\
 & \le \lvert\distributions[n]([r])\rvert f_\beta(S^{\otimes n})  \\
 & \le (n+1)^r f_\beta(S)^n,
\end{split}
\end{equation}
which implies $f_\beta(T)\le f_\beta(S)$ by taking $n$th roots and $n\to\infty$.
\end{proof}

\Cref{thm:exponentsupperset} implies that, for any $f\in\Delta(\graphs,\le)$, the set of admissible exponents $\alpha$ is either empty or an interval that is unbounded from above. For example, it follows from the result of \cite{duan2012zero} that for the Lov\'asz number $f=\vartheta$ \cite{lovasz1979shannon} the exponent $\alpha=2$ gives a monotone extension, therefore every exponent in $[2,\infty)$ is admissible. Note that the extension $\tilde{\vartheta}$ constructed in \cite{duan2012zero}, by virtue of also providing an upper bound on the entanglement-assisted capacity, can be much larger than the Shannon capacity: e.g. $\Theta(\idealquantum{d})=d<d^2=\tilde{\vartheta}$ precisely because the exponent is $2$. What we apparently learn from \cref{thm:exponentsupperset} is that there exist \emph{even worse} upper bounds on the Shannon capacity that extend $\vartheta$. A more optimistic viewpoint is that extensions with large exponents may imply \emph{better} upper bounds on the zero-error \emph{quantum} capacity. In addition, we will see below by an explicit calculation that the set of admissible exponents for $\vartheta$ is in fact $[1,\infty)$.

\subsection{Concrete parameters}

The aim of this section is to find or estimate the set of admissible exponents for specific elements of $\Delta(\graphs,\le)$. To simplify notation, we will use a slight extension of the ``adjacent or equal'' relation $\adjacentorequal$: considering the noncommutative graph $\bigoplus_{d=1}^r\graphnc{G_d}\otimes\idealquantum{d}$ and elements $g\in G_d$, $g'\in G_{d'}$, we will write $g\adjacentorequal g'$ if $d=d'$ and $g$ is adjacent or equal to $g'$ in $G_d$ (effectively regarding $\bigsqcup_{d=1}^r V(G_d)$ as the vertex set of the disjoint union $\bigsqcup_{d=1}^r G_d$, although this graph will not enter the discussion otherwise). In addition, we will introduce the map $\pi:\bigsqcup_{d=1}^r V(G_d)\to[r]$, sending $g\in G_d$ to $d$, for each such direct sum (the map $\pi$ depends on the noncommutative graph, but hopefully there is no risk of confusion arising from not including it in the notation).

We start with a lemma that allows us to bring a generic cohomomorphism between the elements of $\mathcal{A}$ to a simpler form, which can be described via a function between the vertex sets together with a suitable family of isometries. For the purposes of this section, it would be sufficient to prove the statement in the special case when $T$ is a classical graph. The general case will be used in \cref{sec:convexity}.
\begin{lemma}\label{lem:specialcohomomorphism}
Let $S=\bigoplus_{d=1}^r\graphnc{G_d}\otimes\idealquantum{d}$ and $T=\bigoplus_{d=1}^r\graphnc{H_d}\otimes\idealquantum{d}$ be two arbitrary elements of $\mathcal{A}$. Then the following are equivalent:
\begin{enumerate}
\item\label{it:cohomomorphismexists} $T\le S$,
\item\label{it:specialcohomomorphism} there is a function $\varphi:\bigsqcup_{d=1}^r V(H_d)\to\bigsqcup_{d=1}^r V(G_d)$ and a family of isometries $U_h:\complexes^{\pi(h)}\to\complexes^{\pi(\varphi(h))}$ ($h\in \bigsqcup_{d=1}^r V(H_d)$) such that $h\not\adjacentorequal h'\implies(\varphi(h)\not\adjacentorequal\varphi(h')\text{ or }U_h^*U_{h'}=0)$ and $(h\adjacentorequal h'\text{ and }\varphi(h)\adjacentorequal\varphi(h'))\implies U_h^*U_{h'}=cI_{\pi(h)}$ for some $c\in\complexes$.
\end{enumerate}
\end{lemma}
\begin{proof}
\ref{it:cohomomorphismexists}$\implies$\ref{it:specialcohomomorphism}:
The inequality $T\le S$ means that there exist linear maps $E_i:\bigoplus_{d=1}^r\complexes^{V(H_d)}\otimes\complexes^d\to\bigoplus_{d=1}^r\complexes^{V(G_d)}\otimes\complexes^d$ ($i\in I$) such that $\sum_{i\in I}E_i^*E_i=I$ and for all $i,j\in I$ we have $E_i^*SE_j\subseteq T$.
Let $h\in \bigsqcup_{d=1}^r V(H_d)$. Then
\begin{equation}
I_{\pi(h)}=\braket{h}{h}\otimes I_{\pi(h)}=\sum_{g\in\bigsqcup_{d=1}^r V(G_d)}\sum_{i\in I}(\bra{h}\otimes I_{\pi(h)})E_i^*(\ketbra{g}{g}\otimes I_{\pi(g)})E_i(\ket{h}\otimes I_{\pi(h)}),
\end{equation}
therefore there exists $g,i$ such that the corresponding term does not vanish. Choose one such $g$ and $i$ for each $h$, and let $\varphi(h)=g$ and
\begin{equation}
U_h=\frac{(\bra{g}\otimes I_{\pi(g)})E_i(\ket{h}\otimes I_{\pi(h)})}{\norm{(\bra{g}\otimes I_{\pi(g)})E_i(\ket{h}\otimes I_{\pi(h)})}}
\end{equation}
Since $\ketbra{g}{g}\otimes I_{\pi(g)}\in S$, the condition $E_i^*SE_i\subseteq T$ implies that $U_h^*U_h\in(\bra{h}\otimes I_{\pi(h)})T(\ket{h}\otimes I_{\pi(h)})=\complexes I_{\pi(h)}$, i.e. $U_h$ is proportional to an isometry. The normalization ensures that $U_h$ itself is an isometry.

Let $h,h'\in \bigsqcup_{d=1}^r V(H_d)$ and $i,i'\in I$ be the corresponding indices. If $h\not\adjacentorequal h'$ and $\varphi(h)\adjacentorequal\varphi(h')$, then using $\ketbra{\varphi(h)}{\varphi(h')}\otimes I_{\pi(h)}\in S$ and $E_i^*SE_{i'}\subseteq T$ we get $U_h^*U_{h'}\in(\bra{h}\otimes I_{\pi(h)})T(\ket{h'}\otimes I_{\pi(h')})=0$. On the other hand, if $h\adjacentorequal h'$ and $\varphi(h)\adjacentorequal\varphi(h')$, then by the same reasoning we have $U_h^*U_{h'}\in(\bra{h}\otimes I_{\pi(h)})T(\ket{h'}\otimes I_{\pi(h')})=\complexes I_{\pi(h)}$.

\ref{it:specialcohomomorphism}$\implies$\ref{it:cohomomorphismexists}:
Given $\varphi$ and $U_h$ as in the statement, we can choose $I=\bigsqcup_{d=1}^r V(H_d)$ as the index set and the Kraus operators $E_h=\ketbra{\varphi(h)}{h}\otimes U_h$. Then
\begin{equation}
\sum_{h\in\bigsqcup_{d=1}^r V(H_d)}E_h^*E_h=\sum_{h\in\bigsqcup_{d=1}^r V(H_d)}\ketbra{h}{h}\otimes U_h^*U_h=\sum_{h\in\bigsqcup_{d=1}^r V(H_d)}\ketbra{h}{h}\otimes I_{\pi(h)}=I
\end{equation}
and
\begin{equation}
E_h^*SE_{h'}=\begin{cases}
\complexes\ketbra{h}{h'}\otimes U_h^*U_{h'} & \text{if $\varphi(h)\adjacentorequal\varphi(h')$}  \\
0 & \text{if $\varphi(h)\not\adjacentorequal\varphi(h')$}.
\end{cases}
\end{equation}
In the first case, either $h\adjacentorequal h'$, in which case $\complexes\ketbra{h}{h'}\otimes U_h^*U_{h'}=\complexes\ketbra{h}{h'}\otimes I_{\pi(h)}\in T$, or $h\not\adjacentorequal h'$, and then the condition ensures $U_h^*U_{h'}=0$.
\end{proof}

In the special case when $T$ is a graph (i.e. $H_d$ is empty unless $d=1$), the isometries $U_h$ are essentially unit vectors $u_h\in\complexes^{\pi(\varphi(h))}$, and the condition is that $h\not\adjacentorequal h'\implies(\varphi(h)\not\adjacentorequal \varphi(h')\text{ or }\braket{u_h}{u_{h'}}=0)$.

\paragraph{Fractional clique cover number.}

Shannon introduced the fractional clique cover number $\complement{\chi_f}$ as an upper bound on the zero-error capacity \cite{shannon1956zero}. This parameter is the largest element of the asymptotic spectrum of graphs \cite{zuiddam2019asymptotic} (in abstract terms, this amounts to saying that it is equal to the asymptotic rank in the semiring of graphs, see \cite[Theorem
3.8]{strassen1988asymptotic} and \cite[Corollary 2.13]{zuiddam2018algebraic}). 

We now show that, interestingly, this fundamental upper bound does not extend to noncommutative graphs (in a compatible way with the preordered semiring structure). Recall that the orthogonal rank $\xi(G)$ of a graph $G$ is the smallest dimension $d$ such that it is possible to map $V(G)$ to the set of unit vectors of $\complexes^d$ in such a way that adjacent vertices get mapped to orthogonal ones.

\begin{proposition}\label{prop:fccnoextensions}
$\complement{\chi_f}$ does not extend to an element of $\Delta(\mathcal{A},\le)$.
\end{proposition}
\begin{proof}
Suppose that $\alpha\ge 1$ is an admissible exponent for $\complement{\chi_f}$. Then for any graph $G$ with at least two distinct nonadjacent vertices (so that $\complement{\xi}(G)>1$), applying the extension to the inequality $\graphnc{G}\le\idealquantum{\complement{\xi}(G)}$ \cite[Theorem 12.]{stahlke2015quantum} gives
\begin{equation}\label{eq:fccalphalowerbound}
\frac{\log\complement{\chi_f}(G)}{\log\complement{\xi}(G)}\le\alpha.
\end{equation}

For $n\in\naturals$ a multiple of $4$, consider the Hadamard graph $\Omega_n$ \cite{ito1985hadamard1,ito1985hadamard2} whose vertex set is $\{-1,1\}^n$, and two such vectors form an edge iff they are orthogonal. Then clearly $\xi(\Omega_n)\le n$. On the other hand, Frankl and R\"odl \cite{frankl1987forbidden} showed that $\alpha(\Omega_n)\le(2-\epsilon)^n$ for some $\epsilon>0$ and all sufficiently large $n$ (see also \cite{frankl1986orthogonal} for the precise value when $n/4$ is a power of an odd prime), which implies $\chi_f(\Omega_n)\ge\frac{\lvert V(\Omega_n)\rvert}{\alpha(\Omega_n)}\ge\frac{2^n}{(2-\epsilon)^n}$ (see \cite[Proposition 3.1.1]{scheinerman2011fractional} for the first inequality). Choose $G=\complement{\Omega_{4k}}$ ($k\in\positiveintegers$) in the estimate \eqref{eq:fccalphalowerbound} and take the limit $k\to\infty$ to get
\begin{equation}
\begin{split}
\alpha
 & \ge \limsup_{k\to\infty}\frac{\log\chi_f(\Omega_{4k})}{\log\xi(\Omega_{4k})}  \\
 & \ge \limsup_{k\to\infty}\frac{4k\log\frac{2}{2-\epsilon}}{\log(4k)}=\infty,
\end{split}
\end{equation}
a contradiction.
\end{proof}

\paragraph{Lov\'asz number.}

The vectors $(u_g)_{g\in V(G)}$ form an \emph{orthonormal representation} of the graph $G$ if each $u_g$ is a unit vector in a complex inner product space $\mathcal{H}$, and $g\not\adjacentorequal g'$ implies $\braket{u_g}{u_{g'}}=0$. The Lov\'asz $\vartheta$ number of the graph is defined as \cite{lovasz1979shannon}
\begin{equation}
\vartheta(G)=\min_{(u_g)_{g\in V(G)},c}\max\frac{1}{\lvert\braket{c}{u_g}\rvert^2},
\end{equation}
where the minimization is over orthonormal representations of $g$ and unit vectors $c$ in $\mathcal{H}$ (the minimum is attained since we may assume $\dim\mathcal{H}\le\lvert V(G)\rvert$ by restricting to the span of the vectors). $\vartheta$ is an element of $\Delta(\graphs,\le)$ \cite{zuiddam2019asymptotic} (see e.g. the review \cite{knuth1994sandwich} for the proofs of the required properties).

Duan, Severini, and Winter defined a quantum version $\tilde{\vartheta}$ of the Lov\'asz $\vartheta$ number \cite{duan2012zero}, which satisfies $\tilde{\vartheta}(\idealquantum{d})=d^2$, and is a monotone semiring-homomorphism $\ncgraphs\to\nonnegativereals$, therefore any exponent $\alpha\ge 2$ is admissible for $\vartheta$ by \cref{thm:exponentsupperset}. We improve this bound to $1$, i.e. show that the set of admissible exponents is $[1,\infty)$.

\begin{proposition}\label{prop:thetaoneadmissible}
The set of admissible exponents for $\vartheta$ is $[1,\infty)$.
\end{proposition}
\begin{proof}
By \cref{thm:exponentsupperset} it suffices to prove that the exponent $1$ is admissible. Suppose that $\bigoplus_{d=1}^r\graphnc{H_d}\otimes\idealquantum{d}\le\bigoplus_{d=1}^r\graphnc{G_d}\otimes\idealquantum{d}$. Let $H=\bigsqcup_{d=1}^r H_d\strongproduct\complement{K_d}$. Since $\idealclassical{d}\le\idealquantum{d}$, we also have $\graphnc{H}\le\bigoplus_{d=1}^r\graphnc{G_d}\otimes\idealquantum{d}$. By \cref{lem:specialcohomomorphism} there exists a function $\varphi:V(H)\to\bigsqcup_{d=1}^r V(G_d)$ and unit vectors $u_h\in\complexes^{\pi(\varphi(h))}$ such that $h\not\adjacentorequal h'\implies(\varphi(h)\not\adjacentorequal\varphi(h')\text{ or }\braket{u_h}{u_{h'}}=0)$.

Let $(v_g)_{g\in G_d}$ be orthonormal representations of $G_d$ in $\complexes^{n_d}$ and $c_d\in\complexes^{n_d}$ a unit vectors ($d\in[r]$) such that
\begin{equation}
\vartheta(G_d)=\max_{g\in V(G_d)}\frac{1}{\lvert\braket{c_d}{v_g}\rvert^2}.
\end{equation}
Combining ideas from the proof that $\vartheta\le\xi$ and that $\vartheta$ is additive under the disjoint union (see \cite[Theorem 11]{lovasz1979shannon} and \cite[18.]{knuth1994sandwich}), we form an orthonormal representation of $H$ in $\bigoplus_{d=1}^r\complexes^{n_d}\otimes\complexes^d\otimes\complexes^d$ as
\begin{equation}
w_h=v_{\varphi(h)}\otimes u_h\otimes\overline{u_h}\in\complexes^{n_{\pi(\varphi(h))}}\otimes\complexes^{\pi(\varphi(h))}\otimes\complexes^{\pi(\varphi(h))}\subseteq\bigoplus_{d=1}^r\complexes^{n_d}\otimes\complexes^d\otimes\complexes^d,
\end{equation}
and consider the unit vector
\begin{equation}
c=\bigoplus_{d=1}^r\sqrt{\lambda_d}c_d\otimes\frac{1}{\sqrt{d}}\sum_{i=1}^d\ket{i}\otimes\ket{i}
\end{equation}
with convex weights
\begin{equation}
\lambda_d=\frac{d\vartheta(G_d)}{\sum_{j=1}^r j\vartheta(G_j)}.
\end{equation}
To see that $(w_h)_{h\in V(H)}$ is indeed an orthonormal representation, we use that $h\not\adjacentorequal h'$ implies either $\varphi(h)\not\adjacentorequal\varphi(h')$ and therefore $\braket{v_{\varphi(h)}}{v_{\varphi(h')}}=0$, or $\braket{u_h}{u_{h'}}=0$. In both cases
\begin{equation}
\braket{w_h}{w_h'}=\braket{v_{\varphi(h)}}{v_{\varphi(h')}}\lvert\braket{u_h}{u_{h'}}\rvert^2=0.
\end{equation}

For this orthonormal representation and unit vector we have
\begin{equation}
\begin{split}
\lvert\braket{c}{w_h}\rvert^2
 & = \lambda_{\pi(\varphi(h))}\left\lvert\braket{c_{\pi(\varphi(h))}}{v_{\varphi(h)}}\right\rvert^2\frac{1}{\pi(\varphi(h))}\left(\sum_{i=1}^{\pi(\varphi(h))}\lvert\braket{i}{u_h}\rvert^2\right)^2  \\
 & \ge\frac{\lambda_{\pi(\varphi(h))}}{\vartheta(G_{\pi(\varphi(h))})\pi(\varphi(h))}  \\
 & = \frac{1}{\sum_{j=1}^r j\vartheta(G_j)},
\end{split}
\end{equation}
therefore
\begin{equation}
\sum_{d=1}^r\vartheta(H_d)d=\vartheta(H)\le\sum_{d=1}^r d\vartheta(G_d).
\end{equation}
This implies that the extension with exponent $1$ is monotone.
\end{proof}

There are no explicit elements of $\Delta(\ncgraphs,\le)$ known that extend $\vartheta$ with exponent $1$. A possible such extension is the parameter $\hat{\theta}$ introduced by Boreland, Todorov, and Winter \cite{boreland2021sandwich}, which is known to be submultiplicative and monotone \cite[Proposition 5.4. and Proposition 6.1.]{boreland2021sandwich}.

\paragraph{Fractional Haemers bounds.}

In \cite{haemers1979some} Haemers found a new upper bound on the Shannon capacity of graphs, which is incomparable with the Lov\'asz number. A fractional version of Haemers' bound \cite{blasiak2013graph,bukh2018fractional} over any field belongs to the asymptotic spectrum of graphs \cite{zuiddam2019asymptotic}. Of the several alternative formulations we will use the following one, which \cite{bukh2018fractional} attributes to Schrijver. An $a/b$-\emph{subspace representation} of a graph $G$ over a field $\mathbb{F}$ is a collection of subspaces $S_g\subseteq\mathbb{F}^a$ ($g\in V(G)$) such that for all $g\in V(G)$ we have $\dim S_g=b$ and
\begin{equation}
S_g\cap\Big(\sum_{\substack{g'\in V(G)  \\  g'\not\adjacentorequal g}}S_{g'}\Big)=\{0\}.
\end{equation}
The \emph{value} of an $a/b$-subspace representation is the number $\frac{a}{b}$, and the \emph{fractional Haemers bound} $\mathcal{H}^{\mathbb{F}}_f(G)$ is the infimum of the values of all subspace representations of $G$ over $\mathbb{F}$. 

Note that if $(S_g)_{g\in V(G)}$ is a subspace representation and $m\in\positiveintegers$, then $(S_g\otimes\mathbb{F}^m)_{g\in V(G)}$ is also a subspace representation with the same value. In particular, any finite collection of subspace representations can be modified to have equal denominators without affecting their values.

It is known that fields with different nonzero characteristic give rise to distinct parameters \cite[Theorem 2.]{bukh2018fractional}. On the other hand, no separation is known for fields of equal characteristic. In fact, \cite[Proposition 24]{li2020quantum} shows that $\mathcal{H}^\reals_f=\mathcal{H}^\complexes_f$. We do not know the set of admissible exponents for the fractional Haemers bounds over all fields, but we will show that over $\complexes$ every exponent in $[1,\infty)$ is admissible, while over certain fields we can exclude small exponents.

\begin{proposition}\label{prop:fHaemersConeadmissible}
The set of admissible exponents for $\mathcal{H}^\complexes_f$ is $[1,\infty)$.
\end{proposition}
\begin{proof}
In the same way as in the proof of \cref{prop:thetaoneadmissible}, suppose $\bigoplus_{d=1}^r\graphnc{H_d}\otimes\idealquantum{d}\le\bigoplus_{d=1}^r\graphnc{G_d}\otimes\idealquantum{d}$ and let $H=\bigsqcup_{d=1}^r H_d\strongproduct\complement{K_d}$, and choose a function $\varphi:V(H)\to\bigsqcup_{d=1}^r V(G_d)$ and unit vectors $u_h\in\complexes^{\pi(\varphi(h))}$ such that $h\not\adjacentorequal h'\implies(\varphi(h)\not\adjacentorequal\varphi(h')\text{ or }\braket{u_h}{u_{h'}}=0)$.

For $\epsilon>0$, let $(S_g)_{g\in V(G_d)}$ be $a_d/b$-subspace representations of the graphs $G_d$ with values satisfying $\frac{a_d}{b}\le\mathcal{H}^\complexes_f(G_d)+\epsilon$. We claim that
\begin{equation}
h\mapsto T_h:=S_{\varphi(h)}\otimes\complexes u_h\subseteq\boundeds\left(\bigoplus_{d=1}^r\complexes^{a_d}\otimes\complexes^d\right)
\end{equation}
is a $(\sum_{d=1}^r a_dd)/b$-subspace representation of $H$. Indeed, any element of $T_h$ is of the form $x\otimes u_h$ for some $x\in S_{\varphi(h)}$, therefore a vector in the intersection $T_h\cap\sum_{\substack{h'\in V(H)  \\  h'\not\adjacentorequal h}}T_{h'}$ can be written in two ways as
\begin{equation}
x\otimes u_h=\sum_{\substack{h'\in V(H)  \\  h'\not\adjacentorequal h}}y_{h'}\otimes u_{h'}.
\end{equation}
Apply the map $I_{a_{\pi(\varphi(h))}}\otimes\bra{u_h}$ to get
\begin{equation}
x=\sum_{\substack{h'\in V(H)  \\  \pi(\varphi(h'))=\pi(\varphi(h))  \\  h'\not\adjacentorequal h}}y_{h'}\braket{u_h}{u_{h'}}.
\end{equation}
The conditions $h\not\adjacentorequal h'$ and $\braket{u_h}{u_{h'}}\neq 0$ imply $\varphi(h)\not\adjacentorequal\varphi(h')$, and since $y_{h'}\in S_{\varphi(h')}$ and these subspaces form a subspace-representation of $G_{\pi(\varphi(h))}$, the vector $x$ must be $0$.

It follows that
\begin{equation}
\begin{split}
\mathcal{H}^\complexes_f(H)
 & \le \sum_{d=1}^r \frac{a_d}{b}d  \\
 & \le \sum_{d=1}^r (\mathcal{H}^\complexes_f(G_d)+\epsilon)d  \\
 & = \sum_{d=1}^r \mathcal{H}^\complexes_f(G_d)+\frac{r(r+1)}{2}\epsilon.
\end{split}
\end{equation}
Since $\epsilon>0$ was arbitrary, we also have
\begin{equation}
\sum_{d=1}^r\mathcal{H}^\complexes_f(H_d)d=\mathcal{H}^\complexes_f(H)\le\sum_{d=1}^r\mathcal{H}^\complexes_f(G_d)d
\end{equation}
which expresses the monotonicity of the extension of $\mathcal{H}^\complexes_f$ to $\mathcal{A}$ with exponent $1$.
\end{proof}
There are no known (multiplicative, additive, monotone) extensions of the $\mathcal{H}^\complexes_f$ to the semiring of all noncommutative graphs. It might be possible to ``fractionalize'' the extension of the Haemers bound introduced in \cite{gribling2020haemers} and obtain such an extension with exponent $1$.

We turn to fields of positive characteristic. In \cite{briet2013violating} it was proved that if $p$ is a sufficiently large prime such that there exists a Hadamard matrix of size $4p$, then there exists a graph $G$ such that $\mathcal{H}^{\finitefield{p}}_f(G)<\Theta_*(G)$, the entanglement-assisted zero-error capacity. Li and Zuiddam noted that a strict inequality also holds for the quantum Shannon capacity and, consequently, $\mathcal{H}^{\finitefield{p}}_f$ is not in the asymptotic spectrum of graphs with respect to the quantum cohomomorphism preorder $\le_q$ \cite{li2020quantum}. The proof relied on the following observation: if a graph $G$ satisfies $\xi(G)\le d$ and has $M$ disjoint $d$-cliques, then $\alpha_q(G)\ge M$ \cite[Proposition 7]{briet2013violating}, i.e. $\complement{K_M}\le_q G$. Interestingly, the same assumption has another implication that we will find useful for bounding the admissible exponents for $\mathcal{H}^{\finitefield{p}}_f$.
\begin{lemma}\label{lem:disjointcliques}
Suppose that the graph $G$ satisfies $\xi(G)\le d$ and has $M$ disjoint $d$-cliques. Then $\idealclassical{Md}\le \graphnc{G}\otimes\idealquantum{d}$.
\end{lemma}
\begin{proof}
Let $(u_g)_{g\in V(G)}$ be unit vectors in $\complexes^d$ such that $g\adjacent g'\implies\braket{u_g}{u_{g'}}=0$, and let $\varphi:[Md]\to V(G)$ be an enumeration of the vertices in $M$ disjoint $d$-cliques. Then the map $\varphi$ and the vectors $i\mapsto u_{\varphi(i)}$ give rise to a cohomomorphism from $\idealclassical{Md}$ to $\graphnc{G}\otimes\idealquantum{d}$ via \cref{lem:specialcohomomorphism} (with Kraus operators $E_i=(\ket{\varphi(i)}\otimes u_{\varphi(i)})\bra{i}$).
\end{proof}

\begin{proposition}\label{prop:fHaemersFpbound}
Let $p$ be an odd prime such that there exists a Hadamard matrix of size $4p$, and suppose that the exponent $\alpha$ is admissible for $\mathcal{H}^{\finitefield{p}}_f$. Then
\begin{equation}
\begin{split}
\alpha
 & \ge\frac{\log\binom{4p-1}{2p}-\log\sum_{i=0}^{p-1}\binom{4p-1}{i}-\log(4p-1)}{\log(4p-1)}  \\
 & \ge \frac{(4p-1)\left[h(\frac{1}{2}+\frac{1}{8p-2})-h(\frac{1}{4}+\frac{1}{16p-4})\right]-\log(16p^2-4p)}{\log(4p-1)}.
\end{split}
\end{equation}
\end{proposition}
\begin{proof}
Let $G$ be the graph with vertex set the set of binary strings of length $n=4p-1$ and Hamming weight $(n+1)/2$, and with edge set the set of pairs of vertices with Hamming distance $(n+1)/2$. In \cite[Proof of Lemma 3]{briet2013violating} it was shown that (see \cite[Example 11.1.3]{cover2012elements} for the second inequality)
\begin{equation}
\mathcal{H}^{\finitefield{p}}(G)\le\sum_{i=0}^{p-1}\binom{n}{i}\le 2^{nh(p/n)}.
\end{equation}

The graph $G$ satisfies $\xi(G)\le n$ \cite[Lemma 8]{briet2013violating} and has at leas $\lvert V(G)\rvert/n^2$ disjoint $n$-cliques \cite[Lemma 9]{briet2013violating}, where
\begin{equation}
\lvert V(G)\rvert=\binom{n}{(n+1)/2}\ge\frac{1}{n+1}2^{nh(\frac{1}{2}+\frac{1}{2n})}.
\end{equation}
It follows by \cref{lem:disjointcliques} that $\idealclassical{\lceil\lvert V(G)\rvert/n\rceil}\le G\otimes\idealquantum{n}$. We apply the extension of $\mathcal{H}^{\finitefield{p}}_f$ with exponent $\alpha$ to this inequality to get $\lvert V(G)\rvert/n\le \mathcal{H}^{\finitefield{p}}_f(G)n^\alpha$, and rearrange as
\begin{equation}
\begin{split}
\alpha
 & \ge\frac{\log(\lvert V(G)\rvert)}{\log n}-\frac{\log\mathcal{H}^{\finitefield{p}}_f(G)}{\log n}-1  \\
 & \ge\frac{\log\binom{4p-1}{2p}-\log\sum_{i=0}^{p-1}\binom{4p-1}{i}-\log(4p-1)}{\log(4p-1)}  \\
 & \ge \frac{(4p-1)\left[h(\frac{1}{2}+\frac{1}{8p-2})-h(\frac{1}{4}+\frac{1}{16p-4})\right]-\log(16p^2-4p)}{\log(4p-1)}.
\end{split}
\end{equation}
\end{proof}

The bound grows as $\Omega(p/\ln p)$, therefore it becomes nontrivial for any suitable sufficiently large prime $p$. Numerically evaluating the bound for the first few odd primes, we find the first nontrivial lower bound at $p=17$ (a Hadamard matrix of size $68$ exists by the Payley construction \cite{paley1933orthogonal}): if $\alpha$ is an admissible exponent for $\mathcal{H}^{\finitefield{17}}_f$, then $\alpha>1.16249\ldots$.

\paragraph{Projective rank.}

The projective rank $\xi_f$ is a graph parameter introduced in \cite{mancinska2016quantum}, and can be viewed as a fractional version of the orthogonal rank $\xi$. An $a/b$-\emph{projective representation} (with $a\in\naturals$, $b\in\positiveintegers$) of a graph $G$ is a collection $(P_g)_{g\in V(G)}$ of rank-$b$ projections on $\complexes^a$ satisfying $g\adjacent g'\implies P_gP_{g'}=0$. The \emph{value} of an $a/b$-projective representation is the number $\frac{a}{b}$, and the \emph{projective rank} $\xi_f(G)$ is the infimum of the values of all projective representations of $G$. The complement $\complement{\xi_f}$ is an element of $\Delta(\graphs,\le)$ \cite{zuiddam2019asymptotic}.

If $(P_g)_{g\in V(G)}$ is a projective representation and $m\in\positiveintegers$, then $(P_g\otimes I_m)_{g\in V(G)}$ is also a projective representation with the same value. This means that, similarly to subspace representations, a finite collection of projective representations can be modified to have equal denominators without affecting their values.

\begin{proposition}\label{prop:complementprojectiverankoneadmissible}
The set of admissible exponents for $\complement{\xi_f}$ is $[1,\infty)$.
\end{proposition}
\begin{proof}
Once again we proceed as in the proof of \cref{prop:thetaoneadmissible}: suppose $\bigoplus_{d=1}^r\graphnc{H_d}\otimes\idealquantum{d}\le\bigoplus_{d=1}^r\graphnc{G_d}\otimes\idealquantum{d}$ and let $H=\bigsqcup_{d=1}^r H_d\strongproduct\complement{K_d}$, and choose a function $\varphi:V(H)\to\bigsqcup_{d=1}^r V(G_d)$ and unit vectors $u_h\in\complexes^{\pi(\varphi(h))}$ such that $h\not\adjacentorequal h'\implies(\varphi(h)\not\adjacentorequal\varphi(h')\text{ or }\braket{u_h}{u_{h'}}=0)$.

For $\epsilon>0$, let $(P_g)_{g\in V(G_d)}$ be $a_d/b$-projective representations of the graphs $\complement{G_d}$ with values satisfying $\frac{a_d}{b}\le\complement{\xi_f}(G_d)+\epsilon$. Then
\begin{equation}
h\mapsto P_{\varphi(h)}\otimes\ketbra{u_h}{u_h}\in\boundeds\left(\bigoplus_{d=1}^r\complexes^{a_d}\otimes\complexes^d\right)
\end{equation}
is a projective representation of $\complement{H}$ with rank-$b$ projections, therefore
\begin{equation}
\begin{split}
\complement{\xi_f}(H)
 & \le\sum_{d=1}^r\frac{a_d}{b}d  \\
 & \le\sum_{d=1}^r(\complement{\xi_f}(G_d)+\epsilon)d  \\
 & =\sum_{d=1}^r\complement{\xi_f}(G_d)d+\frac{r(r+1)}{2}\epsilon.
\end{split}
\end{equation}
Since $\epsilon>0$ was arbitrary, we also have
\begin{equation}
\sum_{d=1}^r\complement{\xi_f}(H_d)d=\complement{\xi_f}(H)\le\sum_{d=1}^r\complement{\xi_f}(G_d)d,
\end{equation}
i.e. the exponent $1$ is admissible for $\complement{\xi_f}$.
\end{proof}

To the best of our knowledge, no explicit element of $\Delta(\ncgraphs,\le)$ is known which extends $\complement{\xi_f}$. An extension of the orthogonal rank to noncommutative graphs has been considered in \cite{stahlke2015quantum,levene2018complexity}.

\section{Logarithmic convexity}\label{sec:convexity}

Beyond the specific elements considered in the preceding section, the asymptotic spectrum of graphs contains uncountably many points due to the convexity result proved in \cite{vrana2021probabilistic}. In this section we show that, similarly, $\Delta(\mathcal{A},\le)$ is log-convex. The main tool in \cite{vrana2021probabilistic} was the probabilistic refinement of the elements of $\Delta(\graphs,\le)$, which we briefly recall before introducing a variant that is suitable for elements of $\Delta(\mathcal{A},\le)$. For more information on convexity properties of asymptotic spectra we refer the reader to \cite{wigderson2021asymptotic}.

Let $f\in\Delta(\graphs,\le)$, $G$ a nonempty graph and $Q\in\distributions(V(G))$. Let $(Q_n)_{n\in\naturals}$ be any sequence such that $Q_n\in\distributions[n](V(G))$ and $Q_n\to Q$. Then it can be shown that the limit
\begin{equation}
\lim_{n\to\infty}\sqrt[n]{f(\typegraph{G}{n}{Q_n})}
\end{equation}
exists and is independent of the particular sequence. We denote the value of the limit by $f(G,Q)$, and call the resulting functional on probabilistic graphs the \emph{probabilistic refinement} of $f$. The parameter $f$ may be reconstructed from the probabilistic refinement as $f(G)=\max_{Q\in\distributions(V(G))}f(G,Q)$ (for every nonempty $G$). The meaning of the log-convexity of $\Delta(\graphs,\le)$ is that functions $(G,Q)\mapsto\log f(G,Q)$ arising in this way form a convex set with respect to the pointwise operations. This property in turn follows from the fact that the set of logarithmic probabilistic refinements is characterized by a family of affine inequalities.

We define the analogous functionals for elements of $\mathcal{A}$ by reducing it to the probabilistic refinement of the asymptotic spectrum of graphs.
\begin{definition}
Let $S=\bigoplus_{d=1}^r\graphnc{G_d}\otimes\idealquantum{d}$ and $Q\in\distributions(\bigsqcup_{d=1}^r V(G_d))$. We define the probabilistic refinement of $f_\alpha$ (as defined in \cref{def:falpha}) as
\begin{equation}
\log f_\alpha(S,Q):=\entropy(\pi_*(Q))+\sum_{d=1}^r\pi_*(Q)(d)\left[\log f(G_d,Q_d)+\alpha\log d\right],
\end{equation}
where $\entropy$ is the Shannon entropy and $Q_d$ is the normalization of $Q$ restricted to $V(G_d)$ (if $Q(V(G_d))=0$ then $f(G_d,Q_d)$ is not well-defined, but the corresponding term vanishes anyway).
\end{definition}
A routine calculation shows that $f_\alpha(S)=\max_{Q\in\distributions(\bigsqcup_{d=1}^r V(G_d))}f_\alpha(S,Q)$. Note also that forming a convex combinations of the functionals of the form $(S,Q)\mapsto f_\alpha(S,Q)$ amounts to combining the corresponding functionals $(G,Q)\mapsto\log f(G,Q)$ as well as the exponents $\alpha$ with the same weights, since $\log f_\alpha(S,Q)$ is built from these in an affine way.

Recall that every map $f_\alpha:\mathcal{A}\to\nonnegativereals$ is a homomorphism, therefore we only need to prove that monotonicity is preserved under convex combinations. A key point concerning the monotonicity property in \cite{vrana2021probabilistic} is the fact that if the inequality $H\le G$ holds, then for every $Q\in\distributions(V(H))$ there exists a $P\in\distributions(V(G))$ such that for \emph{all} spectral points $f$ the inequality $f(H,Q)\le f(G,P)$ holds (i.e. $P$ can be constructed independently of $f$). Concretely, if $\varphi:\complement{H}\to\complement{G}$ is a homomorphism, then one can take $P=\varphi_*(Q)$. We use a similar argument for the log-convexity of $\Delta(\mathcal{A},\le)$.
\begin{theorem}\label{thm:Aspectrumlogconvex}
The set of functionals of the form $(S,Q)\mapsto \log f_\alpha(S,Q)$ with $f_\alpha\in\Delta(\mathcal{A},\le)$ is convex.
\end{theorem}
\begin{proof}
Let $f_\alpha,g_\beta\in\Delta(\mathcal{A},\le)$ and $\lambda\in[0,1]$, and define $h\in\Delta(\graphs,\le)$ and $\gamma\in[1,\infty)$ by $\log h(G,P)=\lambda\log f(G,P)+(1-\lambda)\log g(G,P)$ and $\gamma=\lambda\alpha+(1-\lambda)\beta$. We need to show that $h_\gamma$ is monotone.

Let $T=\bigoplus_{d=1}^r\graphnc{H_d}\otimes\idealquantum{d}$ and $S=\bigoplus_{d=1}^r\graphnc{G_d}\otimes\idealquantum{d}$ and suppose that $T\le S$. By \cref{lem:specialcohomomorphism}, $\varphi:\bigsqcup_{d=1}^r V(H_d)\to\bigsqcup_{d=1}^r V(G_d)$ and a family of isometries $U_h:\complexes^{\pi(h)}\to\complexes^{\pi(\varphi(h))}$ ($h\in \bigsqcup_{d=1}^r V(H_d)$) such that $h\not\adjacentorequal h'\implies(\varphi(h)\not\adjacentorequal\varphi(h')\text{ or }U_h^*U_{h'}=0)$ and $(h\adjacentorequal h'\text{ and }\varphi(h)\adjacentorequal\varphi(h'))\implies U_h^*U_{h'}=cI_{\pi(h)}$ for some $c\in\complexes$.

Let $P\in\distributions(\bigsqcup_{d=1}^r V(H_d))$ and consider a sequence $(P_n)_{n\in\positiveintegers}$ such that $P_n\in\distributions[n](\bigsqcup_{d=1}^r V(H_d))$ and $P_n\to P$. Note that if $(h_1,\ldots,h_n)\in\typeclass{n}{P_n}$ then $(\varphi(h_1),\ldots,\varphi(h_n))\in\typeclass{n}{\varphi_*(P_n)}$ and
\begin{equation}
U_{h_1}\otimes\cdots\otimes U_{h_n}:\complexes^{\pi(h_1)}\otimes\cdots\otimes\complexes^{\pi(h_n)}\to\complexes^{\pi(\varphi(h_1))}\otimes\cdots\otimes\complexes^{\pi(\varphi(h_n))},
\end{equation}
therefore the restrictions $\left.\varphi^{\times n}\right|_{\typeclass{n}{P_n}}$ and $(U_{h_1}\otimes\cdots\otimes U_{h_n})_{(h_1,\ldots,h_n)\in\typeclass{n}{P_n}}$ determine (via \cref{lem:specialcohomomorphism}) a cohomomorphism corresponding to the inequality
\begin{multline}
\idealclassical{\lvert\typeclass{n}{\pi_*(P_n)}\rvert}\otimes\bigotimes_{d=1}^r\graphnc{\typegraph{G_d}{n\pi_*(P_n)(d)}{(P_n)_d}}\otimes\idealquantum{\prod_{d=1}^rd^{n\pi_*(P_n)(d)}}  \\
\le \idealclassical{\lvert\typeclass{n}{\pi_*(\varphi_*(P_n))}\rvert}\otimes\bigotimes_{d=1}^r\graphnc{\typegraph{H_d}{n\pi_*(\varphi_*(P_n))(d)}{(\varphi_*(P_n))_d}}\otimes\idealquantum{\prod_{d=1}^rd^{n\pi_*(\varphi_*(P_n))(d)}}.
\end{multline}
More precisely, for every element $(d_1,\ldots,d_n)\in\typeclass{n}{\pi_*(P_n)}$ (corresponding to the vertices of the first factor on the left hand side) one needs to choose an isomorphism between $\complexes^{\pi(h_1)}\otimes\cdots\otimes\complexes^{\pi(h_n)}$ and $\complexes^{\prod_{d=1}^rd^{n\pi_*(P_n)(d)}}$ and similarly on the right hand side. Note that the isometries only need to be compared when $\varphi(h_i)\adjacentorequal\varphi(h'_i)$ for all $i$ (and, in the second condition, assuming further $h_i\adjacentorequal h'_i$ for all $i$), and these imply $\pi(\varphi(h_i))=\pi(\varphi(h'_i))$ (respectively, $\pi(h_i)=\pi(h'_i)$), therefore the same isomorphism is used for both and thus the choice of the isomorphisms does not matter.

We apply monotonicity of $f_\alpha$ and $g_\beta$ to this inequality, take the logarithm, divide by $n$ and let $n\to\infty$ to get
\begin{align}
\log f_\alpha(T,P)\le \log f_\alpha(S,\varphi_*(P))  \\
\log g_\beta(T,P)\le \log g_\beta(S,\varphi_*(P)),
\intertext{which implies}
\log h_\gamma(T,P)\le \log h_\gamma(S,\varphi_*(P))
\end{align}
by taking the convex combination of the two inequalities. The right hand side is at most $\log h_\gamma(S)$, and maximizing the left hand side over $P$ gives $\log h_\gamma(T)$, therefore $h_\gamma(T)\le h_\gamma(S)$ as required.
\end{proof}

\section{Comments}

\begin{enumerate}
\item $\Delta(\mathcal{A},\le)$ can be identified with the set of pairs $(f,\alpha)\in\Delta(\graphs,\le)\times\reals$ where $\alpha$ is an admissible exponent for $f$, which is the epigraph of an extended real-valued function by \cref{thm:exponentsupperset}. This function is convex by \cref{thm:Aspectrumlogconvex} and lower semicontinuous because $\Delta(\mathcal{A},\le)$ is locally compact \cite{vrana2021generalization}.
\item It is not known if $\Delta(\ncgraphs,\le)$ characterizes the asymptotic preorder between noncommutative graphs similarly to \cref{thm:Strassen}, but it can be shown that $\Delta(\mathcal{A},\le)$ does characterize the asymptotic preorder between elements of $\mathcal{A}$. To see this, we verify the conditions of \cite[Corollary 2]{vrana2021generalization}: if $S,T\in\mathcal{A}$ and $\frac{\ev_T}{\ev_S}$ is bounded, then (specializing to some $f_\alpha$, $\alpha\to\infty$) one can see that the largest $d$ such that $\idealquantum{d}$ appears in $S$ is at least as large as those appearing in $T$; therefore $T\le \graphnc{H}\otimes\idealquantum{d}$ and  $\graphnc{G}\otimes\idealquantum{d}\le S$ for some graphs $H$ and $G$; since the preorder on $\graphs$ is Strassen, we have $H\le\complement{K_r}\strongproduct G$ for some $r\in\naturals$, which implies $T\le\idealclassical{r}\otimes S$.
\item We note that the asymptotic preorder of $\ncgraphs$ (or $\mathcal{A}$), restricted to the subsemiring $\graphs$ may not be the same as the asymptotic preorder of $\graphs$. The latter allows a sublinear number of uses of a noiseless classical channel, while the former allows a sublinear number of uses of a noiseless quantum channel when comparing large powers. \item It is a curious fact that we were able to show that $1$ is an admissible exponent precisely for those $f\in\Delta(\graphs,\le)$ which are known to be elements of $\Delta(\graphs,\le_q)$ and, conversely, we were able to show that $1$ is not admissible for those $f$ which are known to be outside $\Delta(\graphs,\le_q)$ \cite{li2020quantum} (and the violated inequalities involve the same graphs). We do not know if this is merely a coincidence or the two questions are related.
\item The statement that no exponent is admissible for $\complement{\xi}_f$ can be considerably generalized in connection with the log-convex structure. Note that the graphs $\complement{\Omega_n}$ are vertex-transitive, therefore their logarithmic evaluation maps are affine on $\Delta(\graphs,\le)$. Let $f\in\Delta(\graphs,\le)$ be arbitrary and $\lambda\in(0,1]$, and let $g$ be the convex combination of $f$ (with weight $\lambda$) and $\complement{\chi_f}$ (with weight $1-\lambda$). Since $f(\complement{\Omega_{4k}})\ge 2$ (there exist distinct non-adjacent vertices), we have $\log g(\complement{\Omega_{4k}})\ge(1-\lambda)4k\log\frac{2}{2-\epsilon}$, which implies (using an estimate similar to \eqref{eq:fccalphalowerbound}) that no exponent for $g$ is admissible.
\item The following problem arises from the results of \cref{prop:thetaoneadmissible,prop:fHaemersConeadmissible,prop:complementprojectiverankoneadmissible}:  construct explicit extensions of $\vartheta,\mathcal{H}^\complexes_f$ and $\complement{\xi_f}$ with exponent $1$ or, even better, \emph{families} of extensions for the entire parameter range $[1,\infty)$.
\end{enumerate}

\section*{Acknowledgement}

This work was partially funded by the National Research, Development and Innovation Office of Hungary via the research grants K124152, KH129601, by the \'UNKP-21-5 New National Excellence Program of the Ministry for Innovation and Technology, the J\'anos Bolyai Research Scholarship of the Hungarian Academy of Sciences, and by the Ministry of Innovation and Technology and the National Research, Development and Innovation Office within the Quantum Information National Laboratory of Hungary.

\bibliography{refs}{}

\end{document}